\documentclass[%
 reprint,
 amsmath,amssymb,
 aps,
]{revtex4-2}
\usepackage[english]{babel}

\usepackage{amsthm} 
\newtheorem{theorem}{Theorem}[section]   
\newtheorem{corollary}{Corollary}[theorem]
\usepackage{qcircuit, braket}
\usepackage[ruled,vlined]{algorithm2e}

\usepackage{graphicx}
\usepackage{subcaption} 
\usepackage[margin=1in]{geometry}

\usepackage{tikz}

\begin{document}

\preprint{APS/123-QED}

\title{Quantum computing on encrypted data with arbitrary rotation gates}

\author{Mohit Joshi}
    \email{joshimohit@bhu.ac.in}
\author{Manoj Kumar Mishra}%
\author{S. Karthikeyan}%
\affiliation{%
Department of Computer Science, \\
Banaras Hindu University, Varanasi, India - 221005
}%


\begin{abstract}
An efficient technique of computing on encrypted data allows a client with limited capability to perform complex operations on a powerful remote server without leaking anything about the input or output. 
Quantum computing provides information-theoretic security to solve such a problem, and many such techniques have been proposed under the premises of half-blind quantum computation.
However, these approaches are based on some universal combination of non-parametric gates ($H$, $S$, $T$, $CX$, $CZ$, and $CCX$). Hybrid quantum-classical algorithms require parametric gates, which, when decomposed into non-parametric gates, inadvertently increase the depth of the circuit and hence the communication rounds. 
We propose an approach for recursive decryption of any parametric gate, $R_z(\theta)$, without prior decomposition; exactly when $\theta=\pm\pi/2^m$ for $m\in \mathbb{Z^{+}}$, and approximately with arbitrary precision $\epsilon$ for any given $\theta$.
We also show that a blind algorithm based on such a technique needs at most $O(\log_2^2(\pi/\epsilon))$ computation steps and communication rounds, while the techniques based on a non-parametric resource set require $O(\ln^{3.97}(1/\epsilon))$ steps. 
We use this approach to propose a half-blind quantum computation protocol to enable efficient computation on encrypted data in the NISQ era.

\end{abstract}

\maketitle

\section{Introduction}

The future quantum internet will be enabled in a distributed setting, much like the classical supercomputing framework. 
The client will most likely be a limited-resource device capable of communicating with a powerful remote quantum device capable of performing complex computation.
This is evident with companies like IBM, Google, Intel, Amazon, and Microsoft providing quantum-as-a-service platforms \cite{azuma2023quantum,fitzsimons_private_2017,schmidt_error-corrected_2024,kumar2019towards}.

This quantum internet will be vastly capable of providing security unachievable in the classical internet \cite{wehner_quantum_2018}.
Childs proposed one such technique that enables a client to securely delegate their complex operations to a remote quantum server \cite{childs2005secure}.
This type of security comes under the premises of blind quantum computation. It is broadly divided into two categories, namely, fully blind quantum computation (FBQC) and half-blind quantum computation (HBQC).

FBQC is performed over a universal resource set that assures the security of data and computation. 
Broadbent \textit{et al.} proposed the first universal resource set, called the brickwork state, based on measurement-based quantum computation (MBQC) \cite{broadbent2009universal}.
Due to their fault-tolerance and verifiability, several alternative resource sets have been proposed using MBQC over the years \cite{giovannetti2013efficient, perez-delgado_iterated_2015}.
However, these protocols require the server to implement highly entangled cluster states \cite{mantri2013optimal,zhang_hybrid_2019,ma_universal_2024,van2024hardware}.
This limits the protocol's applicability for near-future applications.
This has been shown in experimental demonstrations with four-qubit cluster states implemented on optical setups \cite{barz2012demonstration,greganti2016demonstration,huang2017experimental}.

HBQC, on the other hand, promises security to data only.
This premise enables quantum computation on encrypted data for near-future implementation of secure quantum protocols.
Tham \textit{et al.} recently used such a protocol to experimentally demonstrate the application of quantum fully-homomorphic encryption \cite{tham2020experimental,li2024experimental}.
Other applications of such protocols include secure multiparty computation, interactive proofs, and quantum one-time programs \cite{broadbent2013quantum, tham2020experimental,gustiani2021blind,das2022blind, li2024experimental,kapourniotis2025asymmetric}.

These protocols encrypt the data using classically assisted Pauli's $X$ and $Z$ gates.
The decryption is performed using some commutation property of the corresponding gates.
Fisher \textit{et al.} first demonstrated this protocol using an optical setup of polarized beam splitters, half-waveplates, and quarter-waveplates \cite{fisher_quantum_2014}.
Various such protocols have been proposed using some universal combination of non-parametric universal resource set $\{H,S,T,CX,CZ,CCX\}$ \cite{broadbent2015delegating,marshall_continuous-variable_2016,tan_universal_2017,sano2021blind}.

However, modern hybrid quantum-classical algorithms are based on parametric circuits like $R_z$ rotation gates \cite{
callison2022hybrid,
khait_variational_2023,
shingu2022variational,
li2024blind}.
This means the circuit needs to be decomposed into non-parametric sets to be implemented using existing HBQC protocols.
This inadvertently increases the size of the blind circuit, which affects the depth and communication rounds of the protocols.

Recently, it has been shown that phase-shifted microwave pulses can implement the arbitrary $R_z$ rotation gates \cite{mckay2017efficient, chen_compiling_2023}.
This can greatly reduce the size of the blind circuit using protocols based on direct $R_z$ gate decryption without the need for decomposition.
Moreover, extensive circuit-based quantum simulation platforms enable rapid prototyping of such protocols \cite{javadi2024quantum, bergholm2022pennylaneautomaticdifferentiationhybrid}.

In this study, we propose the first technique of decryption for an arbitrary rotation gate without the need for prior decomposition, hence reducing the server resources in terms of depth, which in turn substantially reduces the communication overhead required for such schemes.
Our main results are:
\begin{itemize}
    \item We introduce a method of recursive decryption of $R_z(\theta)$; exactly when $\theta=\pm\pi/2^m$ for $m \in \mathbb{Z^{+}}$ and approximately with arbitrary precision $\epsilon$ for any given $\theta$.
    \item We show that any blind approach based on this recursive decryption requires at most $O(\log_2^2(\pi/\epsilon))$ communication rounds, while approaches based on non-parametric gates using Solovey-Kitaev decomposition require $O(\ln^{3.97}(1/\epsilon))$ communication rounds. 
    \item Based on the presented decryption technique of $R_z(\theta)$, we propose a half-blind quantum computing protocol for efficient computing on encrypted data in the NISQ era.
\end{itemize}

The rest of the study is organized as follows: Sec. \ref{sec:prelim} introduces the preliminary of the subject matter. In Sec. \ref{sec:proposed}, we present the technique of recursive decryption of arbitrary $Z$ rotation gates (Sec. \ref{sec:dec_rz}), and the protocol of half-blind quantum computation along with its proof of universality, correctness, and blindness of data (Sec. \ref{sec:proposed_protocol}). In Sec. \ref{sec:example}, a simple implementation example has been presented, and at last, Sec. \ref{sec:conclusion} gives the concluding remarks and future implications of the presented results.

\section{Preliminary}\label{sec:prelim}
The protocols of half-blind quantum computation work by encrypting the data using Pauli's $X$ and $Z$ rotation gates \cite{childs2005secure}.  
The general form of the Child's blindness process can be represented as:
\begin{align}
2^n \mathbb{I} = \sum_{j_1, j_2, ..., j_{2n} \in\{ 0,1\}} \bigg(\bigotimes_{i=1}^n Z_i^{j_{2i}} X_i^{j_{2i-1}}\bigg) | \psi \rangle 
\notag \\
\langle \psi | \bigg(\bigotimes_{i=1}^{n} X_i^{j_{2i-1}} Z_i^{j_{2i}}\bigg) ,
\end{align}
where $j_k \in_r \{0,1\}$ is the $k^{th}$ key out of total $2n$ classical keys used in the protocol for a $n$-qubit system $|\psi\rangle$. $X$ and $Z$ are Pauli's rotation matrices, and $\mathbb{I}$ is the identity matrix of size $n$. This equation shows that nothing except the upper bound on the number of qubits is revealed about the data $\ket{\psi}$.

The client sends this encrypted state to the server, which implements the desired unitary to complete the computation. The client then decrypts the result he obtained from the server using appropriate corrections. 
The correction is a unitary $D$ that operates on the result such that encryption gets nullified as:
\begin{equation}
    U = D \cdot U Z^b X^a,
\end{equation}
where $a,b \in_r \{0,1\}$.

For universal computation, the client should be able to delegate gates from a set of universal resources and correct the result. 
This implies the client should be able to decrypt gates from the set $\{H,S,T,CX,CZ,CCX\}$. 
Client is assumed to have the capability to perform $X$ and $Z$ gates, which are essential for encryption.
The decryption of Clifford gates ($H,S,CX,CZ$) does not require any interaction, as given below \cite{childs2005secure}:
\begin{align}
    H_1(X^a_1 Z^b_1 |\psi\rangle_1) &= X^b_1 Z^a_1 (H_1 |\psi\rangle_1), \label{eq:prelim_h}\\
    P_1(X^a_1 Z^b_1 |\psi\rangle_1) &= X^a_1 Z^{a\oplus b}_1 (P_1 |\psi\rangle_1), \\
   CX_{12} (X^a_1 Z^b_1 X^c_2 Z^d_2 |\psi\rangle_{12}) &= (X^a_1 Z^{b\oplus d}_1) \notag \\
                                                        & \quad (X^{a \oplus c}_2 Z^d_2) (CX_{12} |\psi\rangle_{12}),\\
   CZ_{12} (X^a_1 Z^b_1  X^c_2 Z^d_2 |\psi\rangle_{12}) &= (X^a_1 Z^{b\oplus c}_1) \notag \\
                                                        & \quad (X^c_2 Z^{a \oplus d}_2)(CZ_{12} |\psi\rangle_{12}), \label{eq:prelim_cz}
\end{align}

The decryption of the non-Clifford gate $T$ gate needs an additional ancilla qubit and assistance of $S$ correction as shown below \cite{broadbent2015delegating}: 
\begin{align}
    T(X_1^a Z_1^b |\psi\rangle_1 S_2^y Z_2^d | +\rangle_2) &= S_2^{a \oplus y} X_2^{a \oplus m} \notag \\
                                                            & \quad Z_2^{a(m \oplus y \oplus 1) \oplus b \oplus d \oplus y} T|\psi\rangle_2,
    \label{eq:t_gate_dec}
\end{align}
Here, $m$ is the measurement result from the first qubit.
Also, the $CCX$ gate require the assistance of two $CX$ and one $CZ$ gate as given below \cite{tan_universal_2017}:
\begin{align}
    CCX_{123} (X^a_1 Z^b_1 X^c_2 Z^d_2  X^e_3 Z^f_3 |\psi\rangle_{123}) &= (CX^c_{13}  X^a_1 Z^b_1) \notag \\
                                                                        & \quad (CX^a_{23}  X^c_2 Z^d_2) \notag \\ 
                                                                        & \quad (CZ^f_{12} X^e_3 Z^f_3) \notag \\
                                                                        & \quad (CCX_{123}|\psi\rangle_{123}),
\end{align}

\section{Proposed Protocol}\label{sec:proposed}

We assume the client has the capability to perform gates from the set $\{X, Z, Swap, Measure\}$, and the server has the capability to perform $\{H, CZ, R_z\}$ for universal computation. We can contrast it with ancilla-driven approaches, which require the client to possess the capability of $R_z(\pi/m)$ rotation with $m \in \{0,...7\}$, and MBQC approaches, which require the server to possess the capability to process highly entangled graph states \cite{broadbent2009universal,zhang_hybrid_2019,ma_universal_2024}.

We also assume the client is capable of generating random classical bits with either some classical pseudo-random number generator or using quantum random number generators, which is a standard BQC assumption \cite{childs2005secure}. These client-generated random keys do not need to be directly transmitted outside their private space and only need to help regulate the circuit's quantum gate implementation.

We now present the protocol, by firstly proposing an efficient technique for recursive decryption of an arbitrary $R_z$ gate (Sec. \ref{sec:dec_rz}). We then propose a universal scheme of half-blind quantum computation (Sec. \ref{sec:proposed_protocol}) based on the presented recursive decryption scheme. 

\subsection{Decryption of arbitrary $R_z(\theta)$ gates}\label{sec:dec_rz}
In this subsection, we explain our protocol that enables a client capable of performing $X$ and $Z$ gates to securely delegate the desired rotation of the $R_z(\theta)$ with arbitrary precision $\epsilon$.

It is well established in BQC primitives that the decryption of the $T$ gate requires an additional implementation of the $S$ gate (as evident from Eq. (\ref{eq:t_gate_dec})), which can be generalized to any $R_z(\theta)$ gate requiring $R_z(2\theta)$ for decryption \cite{sano2021blind}. We formally prove this generalization (Theorem \ref{thrm1:rz_gate}) and use it to present a technique of recursive decryption for arbitrary $R_z(\theta)$ gates as: 
\begin{equation}
    R_z(\theta)Z^bX^a = R_z^a(2\theta)X^aZ^bR_z(\theta).
    \label{eq:rz_theta_dec}
\end{equation}

However, for the protocol to work, this recursion must stop. We observe that this recursion will stop exactly when $\theta=\pm\pi/2^m$ for $m \in \mathbb{Z^+}$ (Alogrithm \ref{algo1:dec_rz_integral}) and for $\theta\neq \pm\pi/2^m$, we can decrypt with arbitrary precision $\epsilon$ (Alogirithm \ref{algo2:dec_rz_non_integral}). We describe these scenarios in detail below:

\textbf{Exact Decryption:} If $\theta=\pm\pi/2^m$ for $m \in \mathbb{Z^{+}}$, then the recursive decryption stops exactly at $\theta=\pm\pi/2$, as the base condition of this recursion is an $S=R_z(\pi/2)$ gate, which can be decrypted using only the client available $X$ and $Z$ gates as: 
\begin{align}
    R_z(\pm\pi/2)X^aZ^b &= e^{\mp ia\pi/2}Z^a X^a Z^b R_z(\pm \pi/2).
\label{eq:rz_pi_2_2}
\end{align}

Futhermore, for higher powers of $\pi$, it can be noted that $R_z(m\pi)$, where $m \in \mathbb{Z}$, can be decrypted at client side only as:
\begin{equation}
    R_z(m\pi) = \begin{cases}
        I & \text{if }m \equiv 0 \text{ (mod 2)}, \\
        Z & \text{if }m \equiv 1 \text{ (mod 2)}. \\
    \end{cases}
    \label{eq:rz_mpi}
\end{equation}

Using Eq. (\ref{eq:rz_pi_2_2}) and Eq. (\ref{eq:rz_mpi}), we can decrypt $R_z(\pm\pi/2^m)$ exactly with recursive application of $R_z$ gates (Theorem \ref{thrm2:integral_rz}). However, this procedure inadvertently reveals the value of the encryption key $a$ as the next step in the decryption is dependent on the previous step. To prevent this leakage of information to the server, we use an ancilla qubit and $Swap$ gate to flip the content of the working qubit and ancilla qubit as soon as the \textit{run\_of\_one} (subsequent $a=1$ in recursion) stops, as sketched in Algorithm \ref{algo1:dec_rz_integral}. 
Fig. \ref{fig:dec_rz_integral} shows the circuit of the recursive decryption process $R_z(\pm\pi/2^m)$ with assumption that \textit{run\_of\_one} stops after $k^{th}$ step when $\theta=\pm\pi/2^{k-1}$.
\begin{algorithm}[h]
\caption{Decryption of $R_z(\pm\pi/2^m)$ where $m \in \mathbb{Z^{+}}$}
\label{algo1:dec_rz_integral}

\KwIn{$\ket{\psi}$ and $\theta=\pi/2^m$ where $m\in \mathbb{Z^{+}}$.}
\KwResult{Decrypted state $R_z(\theta)\ket{\psi}$.}
Client generates $2m$ encryption keys $a_i, b_i \in_r \{0,1\}$ for $i \in \{0,\dots,m-1\}$\ and assign \textit{run\_of\_one} $\gets0$\;
Client encrypts input $\ket{\psi}$ using $Z^{b_0}X^{a_0}$ and sends it to server\;
Server computes $R_z(\theta)\ket{\psi}$ and sends result to client\;
Client decrypts $\ket{\psi}$ using $X^{a_0}Z^{b_0}$ and updates $\theta \gets 2\theta$\;

\If{$a_0 = 1$}{
    Client updates $\theta \gets 2\theta$ and  \textit{run\_of\_one} $\gets$ 1\;
}

\For{$k \gets 1$ \KwTo $\log_2(\pi / 2\theta)$}{
    \If{run\_of\_one = 1 \textbf{and} $a_{k-1}=0$}{
        Client applies $Swap$ gate on $\ket{\psi}$ between working and ancillary qubit\;
        \textit{run\_of\_one} $\gets$ 2\;
    }
    Client encrypts $\ket{\psi}$ using $Z^{b_k}X^{a_k}$ and sends to server\;
    Server applies $R_z(\theta)\ket{\psi}$ and returns result\;
    Client decrypts $\ket{\psi}$ using $X^{a_k}Z^{b_k}$ and update $\theta \gets 2\theta$\;
}

\If{\textit{run\_of\_one} $=2$}{
    Client applies $Swap$ gate\;
}

\Return $\ket{\psi}$\;

\end{algorithm}

This recursion only requires $m$ rotation gates, which implies $m$ communication rounds between client and server (see proof Theorem \ref{thrm2:integral_rz}).  
Also note, the expected length for \textit{runs\_of\_one} (subsequent $a=1$ in recursion) is $1$. This implies that the asymptotic complexity of decryption is $O(1)$, which can be used for further optimization. However, this optimization is omitted here for the sake of simplicity.

\begin{figure}
    \centering
\[
\Qcircuit @C=1em @R=.7em {
      & \gate{X^{a_0}Z^{b_0}} & \gate{R_z(\pm\pi/2^m)} & \gate{Z^{b_0}X^{a_0}} 
     & \qw & \qw \gg \quad 
    \\
     &  \qw   & \qw       & \qw 
     & \qw & \qw \gg \quad
}
\]
\[
\Qcircuit @C=1em @R=.7em {
    & \gg & \qw  & 
    \gate{X^{a_1}Z^{b_1}} & \gate{R_z(\pm\pi/2^{m-1})} & \gate{Z^{a_1}X^{b_1}} 
    & \qw & \quad \cdots \quad 
    \\
    & \gg & \qw  & 
    \qw        & \qw                & \qw 
    & \qw & \quad \cdots \quad
}
\]
\[
\Qcircuit @C=1em @R=.7em {
    & \quad \cdots \quad & & 
    \qswap  & \gate{X^{a_k}Z^{b_k}} & \gate{R_z(\pm\pi/2^{k})} & \gate{Z^{a_k}X^{b_k}} & 
    \qw & \quad \cdots \quad & &
    \\
    & \quad \cdots \quad & &
    \qswap \qwx & \qw        & \qw                & \qw &
    \qw & \quad \cdots \quad & &
}
\]
\[
\Qcircuit @C=1em @R=.7em {
    & \quad \cdots \quad & &
    \gate{X^{a_{m-1}}Z^{b_{m-1}}} & \gate{R_z(\pm\pi/2)} & \gate{Z^{a_{m-1}}X^{b_{m-1}}} 
    & \qswap      & \qw 
    \\
    &
     \quad \cdots \quad & &
    \qw        & \qw                & \qw 
    & \qswap \qwx & \qw
}
\]
    \caption{Delegation of $R_z(\pm\pi/2^m)$, where $m \in \mathbb{Z^{+}}$ with recursive decryption of $R_z(\pm\pi/2^{m-1})$ untill $m=1$. Here, the $Swap$ gate is applied when \textit{run\_of\_one} (subsequent $a=1$ in recursion) is exhausted. Note, $\{R_z(\theta)\}$ is implemented by server, while $\{X,Z,Swap\}$ are client implementable gates.}
    \label{fig:dec_rz_integral}
\end{figure}
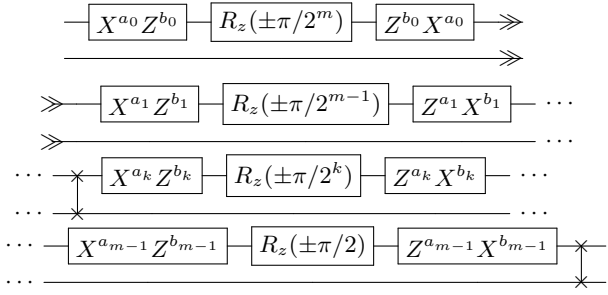

\textbf{Approximate Decryption:}
The exact decryption of $R_z(\pm \pi/2^m)$ for $m \in \mathbb{Z^{+}}$ can be used to approximate any arbitray $\theta$ using a simple observation that:
\begin{equation}
   \theta \approx  p\pi + \sum_{m=1}^M \frac{a_m\pi}{2^m},
   \label{eq:approx_theta}
\end{equation}
with appropriate value of $p \in \mathbb{Z}$ and $a_m \in \{-1,0,1\}$ for $m \in \{1,...,M \}$ where $M=\lceil \log_2(\pi/\epsilon)\rceil$.

Hence, $R_z(\theta)$ can be decrypted using:
\begin{align}
    R_z(\theta) &\approx R_z(p\pi) \prod_{m=1}^M R_z(a_m \pi/2^m).
     \label{eq:rz_decomposed_na}
\end{align}

The procedure of decryption of arbitary $R_z(\theta)$ is sketched in Algorithm \ref{algo2:dec_rz_non_integral}, where $R_z(p\pi)$ is implemented using a simple application of the $Z$ gate (Eq. (\ref{eq:rz_decomposed_na})), and $R_z(a_m\pi/2^m)$ is recursively decrypted using Algorithm \ref{algo1:dec_rz_integral} as a subroutine for all values of $m \in \{1,...,M\}$.

\begin{algorithm}[h]
\caption{Decryption of arbitrary $R_z(\theta)$}
\label{algo2:dec_rz_non_integral}

\KwIn{$\ket{\psi}$ and $\theta$.}
\KwResult{Decrypted state $R_z(\theta)\ket{\psi}$.}

Find coefficients $p \in \mathbb{Z}$ and $a_m \in \{-1,0,1\}$  $\forall m \in \{1,\dots,M\}$,  
where $M \gets \lceil \log_2(\pi / \epsilon) \rceil$,  
such that 
\(
\theta = p\pi + \sum_{m=1}^M \frac{a_m \pi}{2^m};
\)

\If{$p \equiv 1 \pmod{2}$}{
    Client applies $Z$ to $\ket{\psi}$\;
}

\For{$m \gets 1$ \KwTo $\log_2(\pi/\epsilon)$}{
    \If{$a_m = 1$ \textbf{or} $a_m = -1$}{
        $\ket{\psi} \gets$ Call Algorithm~\ref{algo1:dec_rz_integral} with $\ket{\psi}$ and $\theta = a_m\pi/2^m$\;
    }
}

\Return $\ket{\psi}$\;

\end{algorithm}

This recursive decryption requires atmost $O(\log_2^2(\pi/\epsilon))$ rotation gates, which implies atmost $\log_ 2^2(\pi/\epsilon)$ communication rounds (see proof Theorem \ref{thrm3:non_integral_rz}). 
Fig. \ref{fig:complexity_fig} shows the increase in depth between the blind approach based on a parametric $R_z$ gate is $O(\log_2^2(\pi/\epsilon))$, while non-parametric resource set-based blind decryption is $O(\ln^{3.97}(1/\epsilon))$ based on Solovay-Kitaev decomposition (see Corollary \ref{coro1:efficient_rz}). 
Also note, the asymptotic complexity ($O(log_2(\pi/\epsilon))$) makes the recursive decryption with optimization efficient even when modern decomposition algorithms like Ross-Sellinger algorithm (gridsynth) are used, which has asymptotic complexity of $3\log_2(1/\epsilon) + O(\log_2(\log_2(1/\epsilon)))$ \cite{ross2014optimal}. Moreover, if $\theta=\pm\pi/2^m$ for $m \in \mathbb{Z^+}$, then the recursive decryption can be performed in $O(1)$ steps asymptotically. For instance, a setup implementing $R_z(\pi/128)$ using gridsynth will require $104$ $T$ gates asymptotically and $255,856$ $T$ gates using Solovay-Kitaev decomposition for $\epsilon= 10^{-10}$. However, using recursive decryption without decomposition requires only one $R_z$ gate asymptotically, and at most $m = log_2(128) = 7$ gates.

\begin{figure}
    \centering
    \includegraphics[width=0.95\linewidth]{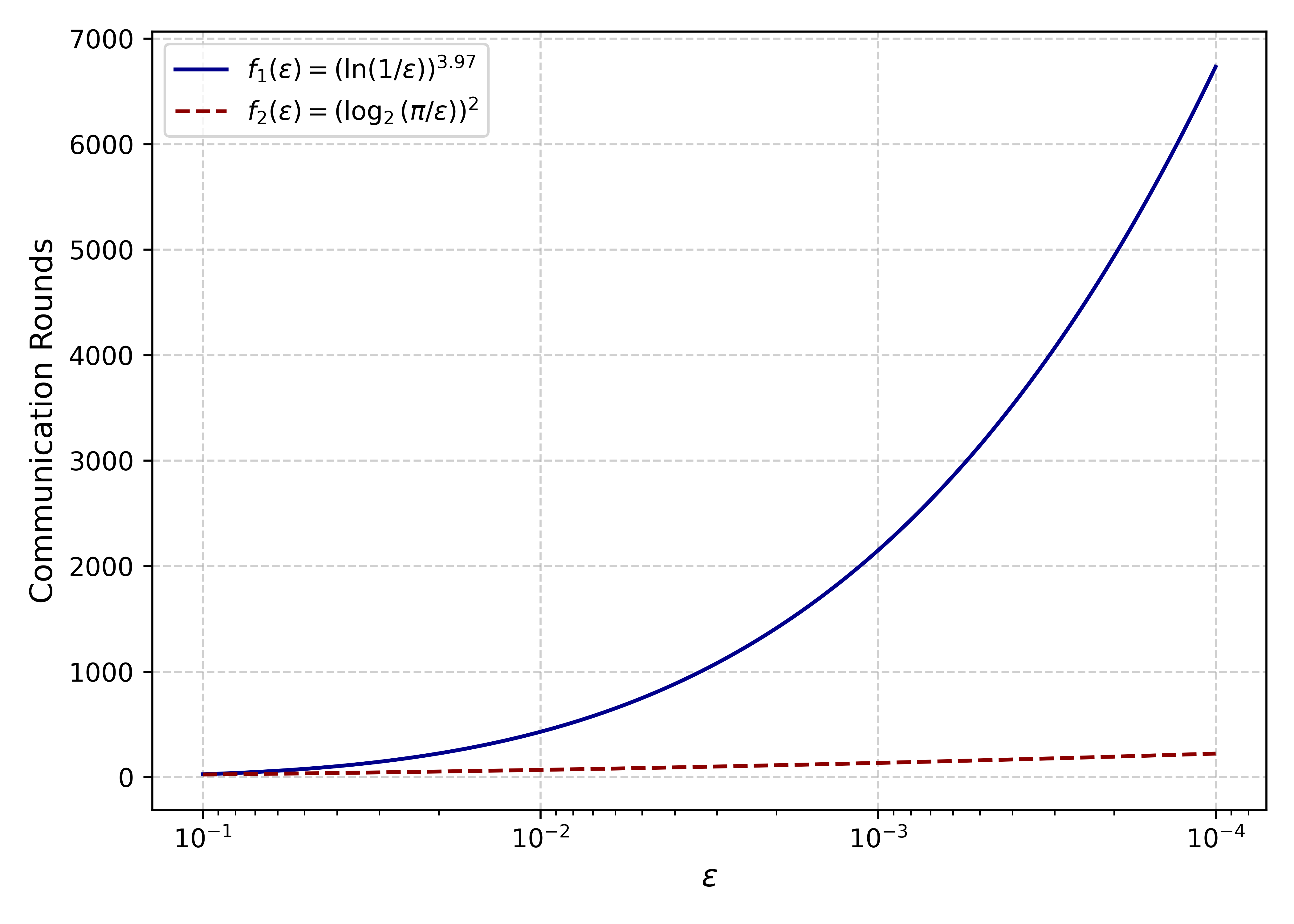}
    \caption{Increase in depth of proposed blind decryption based on arbitrary rotation gate $R_z(\theta)$ ($O(log_2^2(\pi/\epsilon))$) and non-parameteric gates $\{H,S,T\}$  ($O(\ln^{3.97}(1/\epsilon))$) using Solovey-Kitaev decomposition. }
    \label{fig:complexity_fig}
\end{figure}

\begin{algorithm*}
\caption{Proposed Half-Blind Quantum Computation Protocol}
\label{algo3:hdqc_algo}

\KwIn{Algorithm $\mathcal{A}$ as a collection of unitaries $U_q$ with depth $D$, and total number of qubits $n$, 
where $q$ is the ordered set of qubits on which $U$ acts.}

Initialize input $\ket{\psi_0}$\;  
Create $\mathcal{J} \gets \{U_{q,d} \mid d \in \{0, \dots, D'-1\}\}$ from the given algorithm $\mathcal{A}$ with $n' > n$ total number of qubits, including ancilla qubits, and $D'$ total depth of the circuit, including original and trap gates\;

Client sends $\mathcal{J}$ to server over a classical channel\;

\For{$j \gets 0$ \KwTo $D'-1$}{
    Client generates a random key set  
    $\mathcal{K}_j = \{ k_i \mid k_i \in_r \{0,1\},\ i \in \{0,1,\dots,2n'-1\} \}$\;

    Client encrypts previous input:  
    \(
    \ket{\psi}_{j-1,enc} \gets 
    \Big(\bigotimes_{q=0}^{n-1} Z_q^{\mathcal{K}_j[2q+1]} X_q^{\mathcal{K}_j[2q]} \Big) 
    \ket{\psi}_{j-1}
    \)\;

    Client sends $\ket{\psi}_{j-1,enc}$ to server via quantum channel\;

    Server computes  
    $\ket{\psi}_{j,enc} \gets \mathcal{J}_j \ket{\psi}_{j-1,enc}$\;

    Server sends $\ket{\psi}_{j,enc}$ to client\;

    Client decrypts  
    $\ket{\psi}_j \gets dec(\ket{\psi}_{j,enc}, \mathcal{J}_j)$  
    using Eq. (\ref{eq:H_dec}), Eq. (\ref{eq:CZ_dec}), and Algorithm~\ref{algo2:dec_rz_non_integral}\;
}

\Return $\ket{\psi}_{D'-1}$\;

\end{algorithm*}

\subsection{Universal Half-Blind Quantum Computation with arbitrary $R_z(\theta)$ gates}\label{sec:proposed_protocol}
\begin{figure*}
    \centering
    
    \begin{subfigure}[b]{0.45\textwidth}
        \centering
        \includegraphics[width=\linewidth]{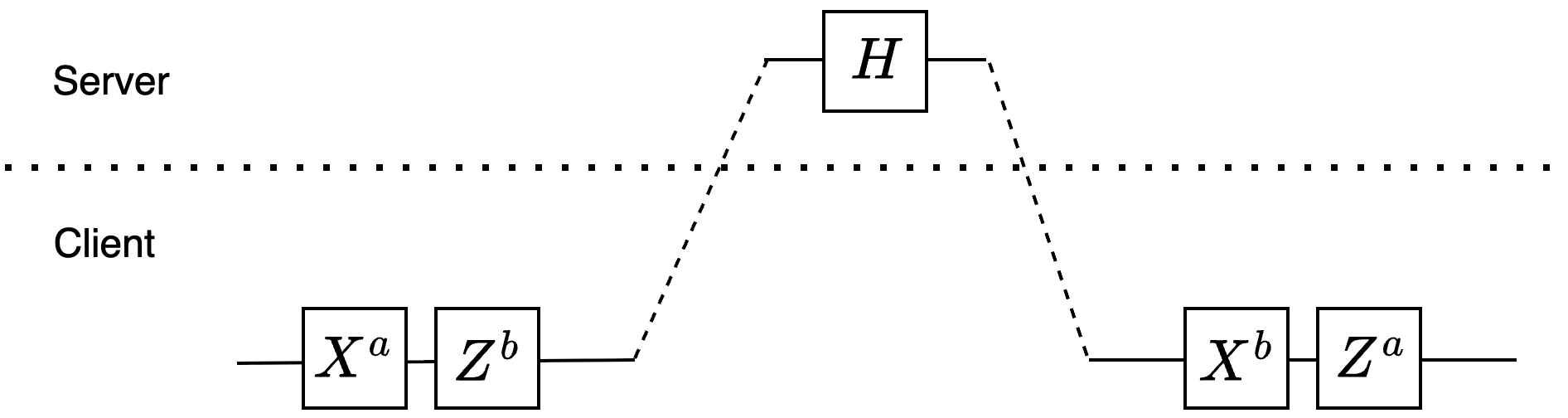}
        \label{fig:first}
    \vspace{0.8em}
    \end{subfigure}
    \begin{subfigure}[b]{0.45\textwidth}
        \centering
        \includegraphics[width=\linewidth]{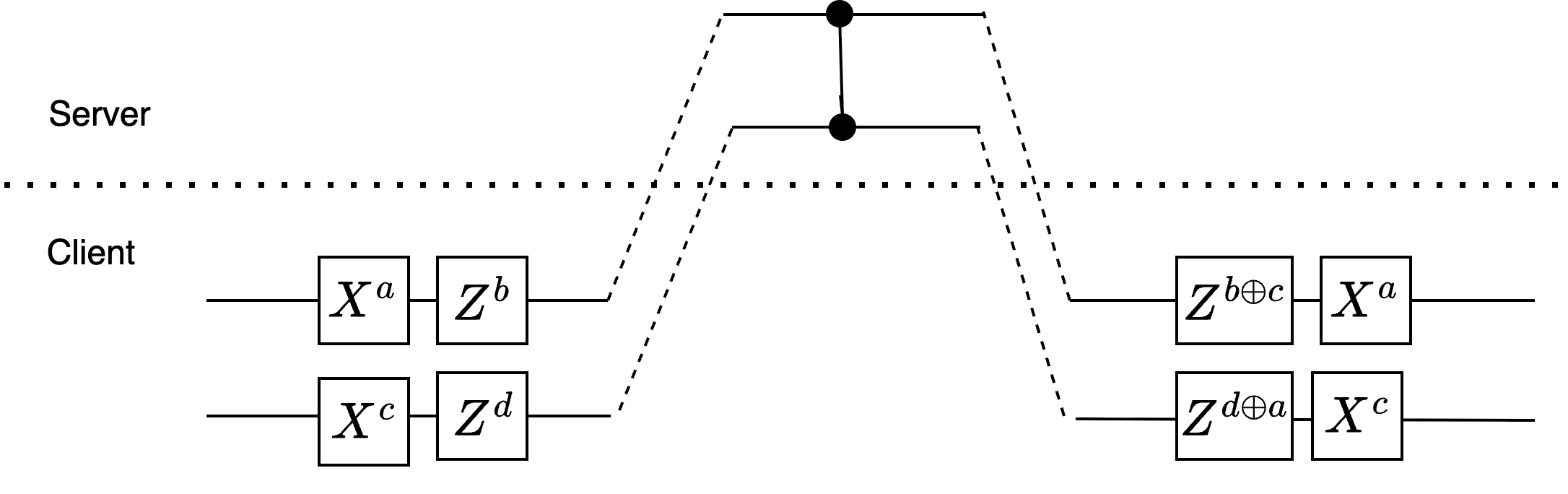}
        \label{fig:second}
    \end{subfigure}
    \begin{subfigure}{0.95\textwidth}
        \centering
        \includegraphics[width=\linewidth]{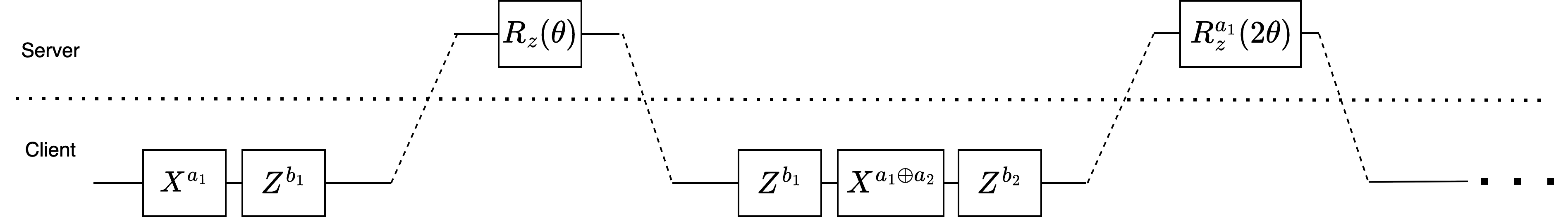}
        \label{fig:third}
    \end{subfigure}
    
    \caption{The process of secure delegation of resources $\{H,CZ,R_z\}$ which is encryption and decrypted by client capable of performing $\{X,Z,Swap,Measure\}$ assited by random classical bit $a_i,b_i \in_r \{0,1\}$.}
    \label{fig:enc_univ_gates}
\end{figure*}

In this section, we propose a half-blind quantum computation protocol based on the presented decryption of arbitrary rotation gates in Section \ref{sec:dec_rz}.

The protocol allows us to securely delegate an algorithm $\mathcal{A}$ given as a collection of unitaries $U_q$ with depth $D$, and total number of qubits $n$, where $q$ denotes the ordered set of qubits on which the unitary $U$ needs to be applied.

The client begins with an input state $|\psi\rangle_0$ defined on $n'>n$-dimensional Hilbert state with at least one ancillary qubit, 
Client uses $\mathcal{A}$ to define a computation set $\mathcal{J} = \{ U_{q,d} | \quad d \in \{0,..., D'-1\}\}$, where $D'$ is the depth of the given circuit with some trap operations. 
For each $j^{th}$ step of $D'$ computation step, which prepares a set of classical random keys $\mathcal{K}_j=\{k_i | k_i \in_r \{0,1\}, i =\{0,1,...,2n'-1\} \}$ of size $2n'$. The total encryption keys needed to complete the $D'$ steps are $(2n'D')$ drawn from the possible binary key space of size $2^{2n'D'}$. 

The protocol runs as an \textit{encrypt-compute-decrypt} cycle for every unitary set at $j^{th}$ depth in computation set $\mathcal{J}$, i.e., $D'$ times, using a quantum channel comprising of working and ancillary qubits, which is a subset of $\{0,...,n'-1\}$.

In the $j^{th}$ cycle of \textit{encrypt} phase, the client encrypts the previous state $|\psi\rangle_{j-1}$ by applying classical controlled $X$ and $Z$ gates on state  using the generated encryption key $\mathcal{K}_j$ as shown in Eq. (\ref{eq:encrypt}):
\begin{equation}
|\psi\rangle_{j-1,enc} = \bigg(\bigotimes_{q=0}^{n-1} X_q^{\mathcal{K}_j[2q]} Z_q^{\mathcal{K}_j[2q+1]}\bigg) |\psi\rangle_{j-1}.
\label{eq:encrypt}
\end{equation}
The client then transmits this encrypted state $\ket{\psi}_{j, enc}$ to the server.

In the $j^{th}$ cycle of \textit{compute} phase, the server applies the unitary $U_{q,j}$ on $|\psi\rangle_{j-1, enc}$ defined in the computation set $\mathcal{J}$ as shown in Eq. (\ref{eq:compute}):
\begin{equation}
    |\psi\rangle_{j,enc} = U_{q,j}|\psi\rangle_{j-1,enc},
    \label{eq:compute}
\end{equation}
where $q$ denotes the qubits on which the computation $U_j$ has to be performed. 
The resultant state $U_j|\psi\rangle_{j, enc}$ is returned to the client. 

In the $j^{th}$ cycle of \textit{decrypt} phase, the client recovers the correct state by decrypting the server output using Eq. (\ref{eq:decrypt}):
\begin{equation}
    U_j|\psi\rangle_j = dec(U_j|\psi\rangle_{j,enc}),
    \label{eq:decrypt}
\end{equation}
where decryption map $dec(\cdot)$ is defined using the keys in encryption set $\mathcal{K}_j$ and the relations given in Eq. (\ref{eq:H_dec}) to Eq. (\ref{eq:Rz_dec}) using Algorithm \ref{algo2:dec_rz_non_integral}.
Fig. \ref{fig:enc_univ_gates} shows the encryption and decryption process of the universal set $\{H,CZ,R_z\}$.
\begin{align}
    H_1 Z_1^b X_1^a  &= X_1^b Z_1^a H \label{eq:H_dec} \\
    CZ_{1,2} Z_2^d X_2^c Z_1^b X_1^a &= Z_2^{d \oplus a } X_2^c Z_1^{b \oplus c} X_1^a CZ_{1,2} \label{eq:CZ_dec}\\
    R_z(\theta)_1 Z_1^b X_1^a &= R_z^a(2\theta)_1 Z_1^b X_1^a R_z(\theta)_1 \label{eq:Rz_dec}
\end{align}

\textbf{Proof of Universality}: 
As a set of Clifford and non-Clifford gates like $\{X,Z,H,S,T,CX\}$ form a universal set for quantum computation \cite{nielsen2001quantum}. It is trivial to show that a server capable of performing only $\{H,CZ,R_z\}$ will be able to let a client with $X$ and $Z$ gates perform universal computation. 

Using Euler's Z-X-Z decomposition, we can represent any single qubit gate as:
\begin{equation}
    U = e^{i\phi} R_z(\alpha)R_x(\beta)R_z(\gamma)
\end{equation}
Also $R_x(\theta) = H R_z(\theta)H $. Hence, any single-qubit gate can be performed using $H, R_z(\theta)$.

Also, $CX_{1,2} = H_2CZ_{1,2}H_2$, which can be used to reach any multi-qubit operation along with non-Clifford resource $T(=R_z(\pi/4))$ gate.

\textbf{Proof of Correctness}:
Correctness of $H$ and $CZ$ directly follows from literature \cite{broadbent2015delegating, tan_universal_2017}, and can be verified for all possible $\ket{\psi}$ and $a,b,c,d \in_r \{0,1\}$ with the Eq. (\ref{eq:H_dec}) and Eq. (\ref{eq:CZ_dec}).
Theorem \ref{thrm1:rz_gate}, Theorem \ref{thrm2:integral_rz}, and Theorem \ref{thrm3:non_integral_rz} proves the correctness of recursive decryption of $R_z(\theta)$ as given in Eq. (\ref{eq:child_rz}).
\begin{align}
       R_z(\theta)\ket{\psi}\bra{\psi}R_z(\theta) &= R^a_z(2\theta) Z^b X^a R_z(\theta)\ket{\psi}
       \notag \\& \quad 
        \bra{\psi} X^a Z^b R_z(\theta) \label{eq:child_rz}
\end{align}

\textbf{Proof of Blindness of Data}:
For the blindness of data, we need to show that the encryption keys do not leak to the server. 
The blindness of Clifford's resources is straightforward, as it does not require any additional interaction between client and server. This can be proven using Child's equation:
\begin{align}
       2\mathbb{I} &=   \sum_{a,b \in \{0,1\}} Z^bX^a H\ket{\psi} \bra{\psi} X^aZ^b H \\     
       4\mathbb{I} &= \sum_{a,b,c,d \in \{0,1\}}Z_2^{d} X_2^c Z_1^{b} X_1^a 
       CZ_{1,2}\ket{\psi} 
       \notag \\& \quad \quad \quad \quad \quad \quad 
       \bra{\psi}  X_1^a Z_1^b X_2^c Z_2^d CZ_{1,2} \\
\end{align}

The proof for the blindness of $R_z(\theta)$ is based on showing that the proposed protocol of recursive decryption, as presented in Section \ref{sec:dec_rz}, can be converted to an equivalent delayed version, a technique of proof similar to Ref. \cite{broadbent2015delegating}. 

We start by noting that the first delegation of the $R_z(\theta)$ gate decryption can be represented as shown in Fig. \ref{fig:proof_fig1}. This can be represented by an equivalent entanglement-based circuit as represented in Fig. \ref{fig:proof_fig2}. 
The delayed measurement of qubits shows that the first delegation of $R_z(\theta)$ can be implemented without revealing the encryption keys $a,b$. 

However, for exact decryption, successive delegation of $R_z(2\theta)$ reveals the previous encryption key (Eq. \ref{eq:rz_pi_2}). Application of the $Swap$ gate and ancillary qubit at the client side makes the recursion steps independent of encryption keys by swapping the content of ancilla and working qubit, as soon as \textit{run\_of\_one} stops. At the server side, this will resemble as if the encryption key comprises all $1$s, without revealing any additional information.
This also ensured the privacy for any arbitrary $R_z$ gate as the approximate decomposition of $\theta$ as given in Eq. \ref{eq:approx_theta} is independent of encryption keys.

This procedure prevents the leakage of the encryption key, hence preserving the blindness of data $\ket{\psi}$.

\begin{figure}
    \centering
    \includegraphics[width=0.95\linewidth]{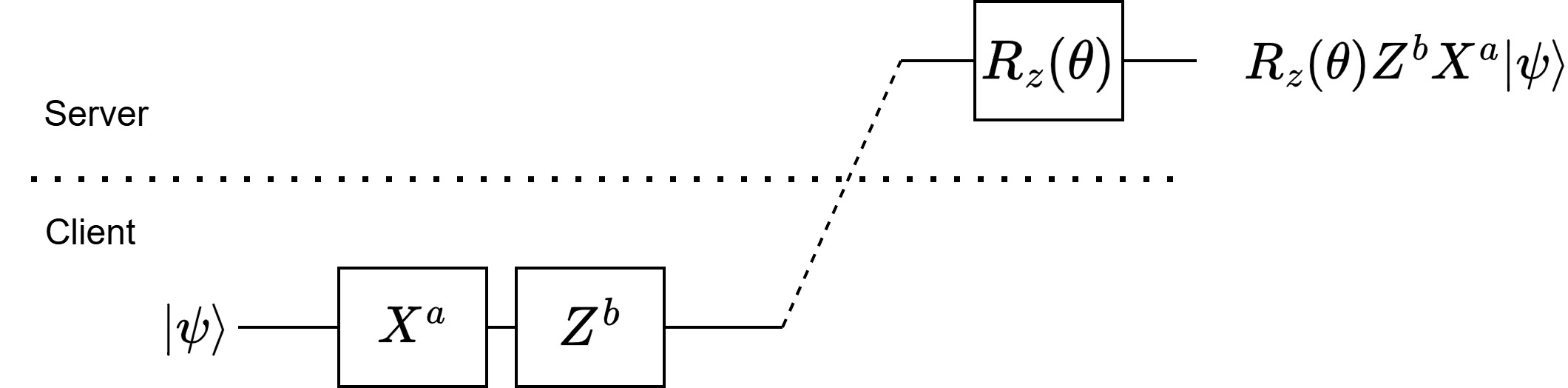}
    \caption{The delegation of $R_z(\theta)$ gate in the proposed protocols.}
    \label{fig:proof_fig1}
\end{figure}

\begin{figure}
    \centering
    \includegraphics[width=0.95\linewidth]{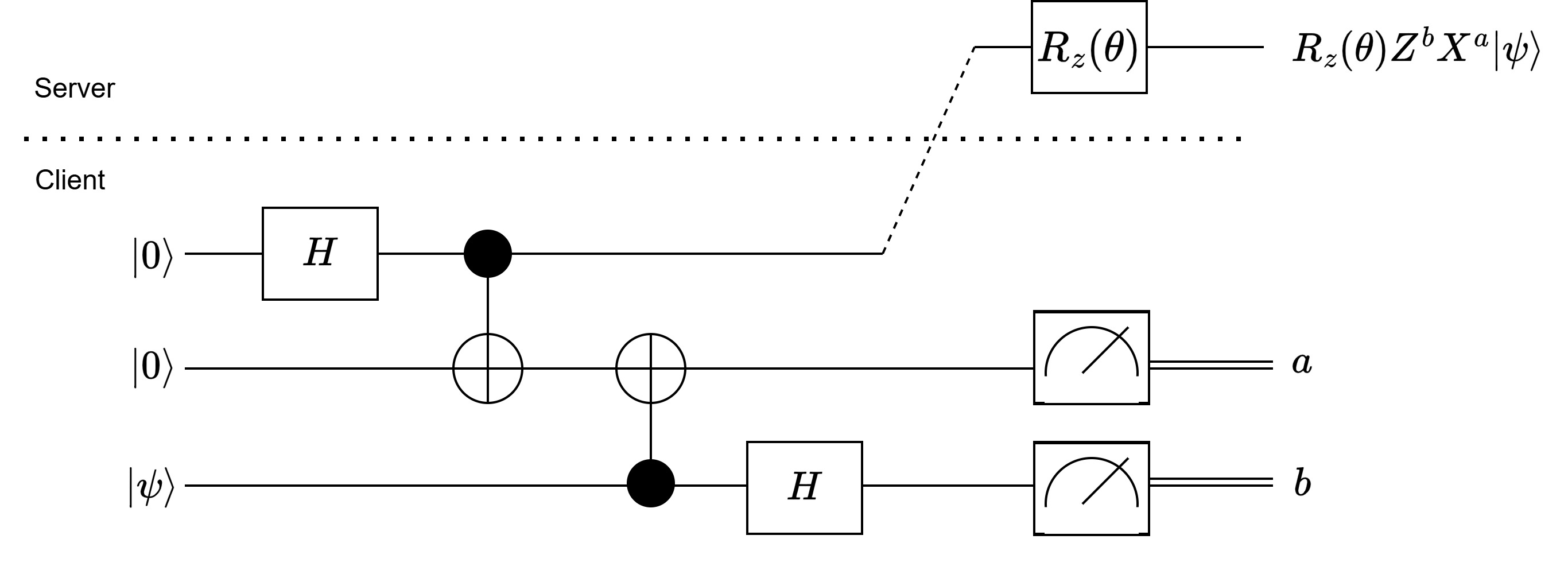}
    \caption{The entanglement-based equivalent of proposed decryption of $R_z(\theta)$ with delayed correction.}
    \label{fig:proof_fig2}
\end{figure}

\section{Example}\label{sec:example}
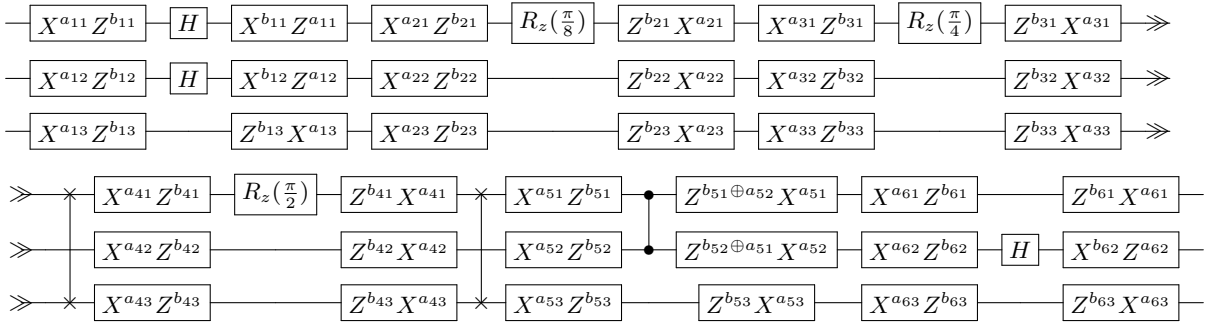
\begin{figure*}[t]
    \centering
\[
\Qcircuit @C=1em @R=.7em {
    & 
    \gate{X^{a_{11}} Z^{b_{11}}} & \gate{H} & \gate{X^{b_{11}} Z^{a_{11}}} & 
    \gate{X^{a_{21}} Z^{b_{21}}} & \gate{R_z(\tfrac{\pi}{8})} & \gate{Z^{b_{21}} X^{a_{21}}} & \gate{X^{a_{31}} Z^{b_{31}}} & 
    \gate{R_z(\tfrac{\pi}{4})} & \gate{Z^{b_{31}} X^{a_{31}}} & 
    \qw & \qw \gg \quad 
    &   
    \\
    & \gate{X^{a_{12}} Z^{b_{12}}} & \gate{H} & \gate{X^{b_{12}} Z^{a_{12}}} & \gate{X^{a_{22}} Z^{b_{22}}} & \qw                        & \gate{Z^{b_{22}} X^{a_{22}}} & \gate{X^{a_{32}} Z^{b_{32}}} & \qw                        & \gate{Z^{b_{32}} X^{a_{32}}} &
    \qw & \qw \gg \quad 
    &   
    \\
    & \gate{X^{a_{13}} Z^{b_{13}}} & \qw      & \gate{Z^{b_{13}} X^{a_{13}}} & \gate{X^{a_{23}} Z^{b_{23}}} & \qw                        & \gate{Z^{b_{23}} X^{a_{23}}} & \gate{X^{a_{33}} Z^{b_{33}}} & 
    \qw                        & \gate{Z^{b_{33}} X^{a_{33}}} & \qw & \qw \gg \quad 
    &
}
\]
\[
\Qcircuit @C=1em @R=.7em {
    & \gg & \qw  & 
     \qswap & \gate{X^{a_{41}} Z^{b_{41}}} & \gate{R_z(\tfrac{\pi}{2})}  &  \gate{Z^{b_{41}} X^{a_{41}}} & \qswap & \gate{X^{a_{51}} Z^{b_{51}}} & \ctrl{1} & \gate{Z^{b_{51} \oplus a_{52}} X^{a_{51}}} & \gate{X^{a_{61}} Z^{b_{61}}} & \qw      & \gate{Z^{b_{61}} X^{a_{61}}} & \qw
    \\
    & \gg & \qw  &  
      \qw & \gate{X^{a_{42}} Z^{b_{42}}}   & \qw                  &  \gate{Z^{b_{42}} X^{a_{42}}}    & \qw    & \gate{X^{a_{52}} Z^{b_{52}}} & \control \qw & \gate{Z^{b_{52} \oplus a_{51}} X^{a_{52}}} & \gate{X^{a_{62}} Z^{b_{62}}} & \gate{H} & \gate{X^{b_{62}} Z^{a_{62}}} & \qw
    \\
    & \gg & \qw  & 
    \qswap \qwx[-2] & \gate{X^{a_{43}} Z^{b_{43}}} & \qw           & \gate{Z^{b_{43}} X^{a_{43}}}             & \qswap \qwx[-2] & \gate{X^{a_{53}} Z^{b_{53}}} & \qw      & \gate{Z^{b_{53}} X^{a_{53}}} & \gate{X^{a_{63}} Z^{b_{63}}} & \qw      & \gate{Z^{b_{63}} X^{a_{63}} } & \qw
}
\]
\caption{Complete blind implementing of circuit given by computation set $\mathcal{J}_1$ in Eq. (\ref{eq:ex_J1}), where $\{X,Z,Swap\}$ gate set is performed by client while $\{H,CZ, R_z\}$ gate set is delegated to remote server. }
    \label{fig:complete_circ}
\end{figure*}

In this section, we show the secure implementation of an example circuit containing $H$, $R_z(\pi/8)$, and $CX$ gates using Algorithm \ref{algo3:hdqc_algo}.

Suppose we are given a two-qubit circuit $\mathcal{C}$ as shown below:
\[
\Qcircuit @C=1em @R=.7em {
    & \gate{H} & \gate{R_z(\tfrac{\pi}{8})} & \ctrl{1} & \qw \\
    & \qw      & \qw                        & \targ    & \qw
}
\]

The expected output for given input $\ket{\psi}_{1,2}=\alpha_1\ket{00}_{1,2} + \alpha_2\ket{10}_{1,2} + \alpha_3\ket{01}_{1,2}+ \alpha_4\ket{11}_{1,2}$ of the circuit $\mathcal{C}$ is:
\begin{align}
    CX_{1,2}R_z(\pi/8)_1H_1\ket{\psi}_{1,2} &= (\alpha_1 + \alpha_3)\ket{00}_{1,2} \notag\\
    & \quad + (\alpha_2 + \alpha_4)\ket{10}_{1,2} \notag\\
    & \quad + (\alpha_2-\alpha_4)e^{i\pi/8}\ket{01}_{1,2} \notag\\
    & \quad + (\alpha_1 -\alpha_3)e^{i\pi/8}\ket{11}_{1,2}
\end{align}

The circuit is first transpiled to given universal resource set $\{H,CZ,R_z\}$ as:

\[
\Qcircuit @C=1em @R=.7em {
    & \gate{H} & \gate{R_z(\tfrac{\pi}{8})} & \ctrl{1} & \qw \\
    & \gate{H} & \qw                        & \control \qw & \gate{H} & \qw
}
\]

Now, we start by creating the computation set $\mathcal{J}_1$ from this circuit which will be:
\begin{align}
    \mathcal{J}_1 & = \{  H_{(0,),0}, H_{(1,),0}, R_z(\pi/8)_{(0,),1}, R_z(\pi/4)_{(0,),2}, \notag\\ 
    & \quad R_z(\pi/2)_{(0,),3}, CZ_{(0,1),4}, H_{(1,),5} \}.
    \label{eq:ex_J1}
\end{align}
Here, $D'=6$, $n'=3$, and the secure implementation requires the key of size $2n'D'= 36$. 

In the Fig. \ref{fig:complete_circ}, we show the complete circuit of client and server side to perform computation on the encrypted state $\ket{\psi}$. The gate set $\{X,Z,Swap\}$ is performed by the client while the gate set $\{H,CZ,R_z\}$ is server implemented.

The protocol starts with the client encrypting the input state $\ket{\psi}$ using $Z_1^{b_{11}}X_1^{a_{11}}$ on first qubit, $Z_2^{b_{12}}X_2^{a_{12}}$, and $Z_3^{b_{13}}X_3^{a_{13}}$ on third qubit. The client then sent these qubits to the server, who then applies $H_1H_2$, which will be decrypted by the client with the appropriate key as given in Eq. (\ref{eq:H_dec}).

Client then again encrypts the resultant state using next round keys $a_{21}a_{22}, a_{23}$ for $X$ gate and $b_{21}, b_{22}, b_{23}$ for $Z$ gate. The $R_z(\pi/8)$ needs two additional recursive call to $R_z(\pi/4)$ and $R_z(\pi/2)$ for decryption. In the chosen example, we have assumed that \textit{run\_of\_one} stops at $a_{31}$ i.e. $a_{11}=1, a_{21}=1,a_{31}=0$. Hence, the $Swap$ gate is used after $R_z(\pi/4)$ rotation. Another $Swap$ gate is applied after the recursive decryption of $R_z(\pi/8)$ has been completed at the client's side to restore the state of the working qubit. After this, $CZ_{1,2}$ and $H_2$ gates are decrypted using Eq. (\ref{eq:CZ_dec}) and Eq. (\ref{eq:H_dec}) to complete the given blind circuit.

\section{Conclusion}\label{sec:conclusion}
In this paper, we present a protocol for recursive decryption of $R_z$ gates with arbitrary rotations. We showed that this recursion requires at most $O(\log_2^2(\pi/\epsilon))$ server resources and communication rounds.
This is in contrast to protocols based on the decryption of non-parametric resources $\{H,S,T,CX,CZ,CCX\}$, which will require $O(\ln^{3.97}(1/\epsilon))$ steps if decomposition is performed using the Solovay-Kitaev theorem. 
Using the recursive decryption of arbitrary rotation gates, we have presented an efficient technique to perform half-blind quantum computation. 
This enables computing on encrypted data with the least amount of entangled resources and communication rounds.
This implies that larger proof-of-principle experiments can be implemented using phase-shifted microwave pulses \cite{mckay2017efficient, chen_compiling_2023}, enabling secure implementation of hybrid quantum-classical algorithms.

Our study considers only a blind approach, while an efficient full blind approach, where $\theta$ is blind, will enhance the practicality of the solution. Moreover, in our approach client still needs the ability to prepare the state $\ket{\psi}$, which can be improved using the assumption of non-communication server \cite{morimae_secure_2013} and entanglement-swapping-based triple server protocols \cite{li_triple-server_2014}. Moreover, verification of the server's honesty, that is, the server performed the operations correctly, is an important question to answer. Although a simple trap-based mechanism can be implemented for verification of the server in our protocol, still more elegant and cost-effective solution can be explored \cite{morimae_measurement-only_2016,fitzsimons_unconditionally_2017}.

\begin{acknowledgments}
This research is supported by a seed grant under the IoE, BHU [grant no. R/Dev/D/IoE/SEED GRANT/2020-21/Scheme No. 6031].

\end{acknowledgments}

\appendix

\section{Decryption of arbitrary $Z$ gate}

\begin{theorem}\label{thrm1:rz_gate}
The decryption of arbitrary $R_z(\theta)$ is dependent on $R_z(2\theta)$, i.e.,
$R_z(\theta)Z^bX^a = R^a_z(2\theta)X^aZ^bR_z(\theta)$.
\end{theorem}

\begin{proof}
For decryption of $R_z$ gate encrypted by $X$ and $Z$ gate, we need to find a unitary $D$ such that:
\begin{equation}
    R_z(\theta) = D \cdot R_z(\theta) Z^b X^a,
\end{equation}
where $a,b \in_r \{0,1\}$.
Here, solving for $D$ we get:
\begin{align}
    D &= R_z(\theta) X^a Z^b R_z(-\theta).
\end{align}

Note that $R_z^\dagger(\theta) = R_z(-\theta)$.
Also,
\begin{equation}
    R_z(\theta) = \begin{pmatrix}
        e^{-i\theta /2} & 0 \\
        0 & e^{i\theta/2}
    \end{pmatrix}.
\end{equation}

Given $X$ and $Z$ are standard Pauli rotation gates, we can represent $X^a$ and $Z^b$ as:
\begin{align}
    X^a = \begin{pmatrix}
        1-1_a & 1_a \\
        1_a & 1-1_a \\
    \end{pmatrix}, &
    Z^b = \begin{pmatrix}
        1 & 0 \\
        0 & (-1)^{1_b}
    \end{pmatrix},
\end{align}

where 
\begin{equation}
1_x = \begin{cases} 1, & \text{if } x =1, \\ 
0, & \text{if } x=0.
\end{cases}
\end{equation}

This indicator variable representation of $X$ and $Z$ gate allows us to algebraically manipulate the matrix form of $D$, which can be represented as:
\begin{equation}
    D = \begin{pmatrix}
        1-1_a & (-1)^{1_b} 1_a e^{-i\theta} \\
        1_a e^{i\theta} & (-1)^{1_b}(1-1_a)
    \end{pmatrix}.
\end{equation}
This matrix representation of $D$ can be further decomposed into simpler gate combinations as:
\begin{align}
   D &= \begin{pmatrix}
        e^{i\theta1_a} & 0 \\
        0 & e^{i\theta1_a}
    \end{pmatrix}
    \begin{pmatrix}
        1-1_a & 1_a \\
        1_a & 1-1_a \\
    \end{pmatrix}
        \begin{pmatrix}
        1 & 0 \\
        0 & (-1)^{1_b}
    \end{pmatrix}, \notag\\
    &= R^a_z(2\theta) X^a Z^b.
\end{align}

Here, we have used the following identities associated with indicator variables $1_a$ and $1_b$, where $a,b \in_r \{0,1\}$:
\begin{align}
    (-1)^{2\cdot 1_b}  &= 1, & (1-1_a)\cdot 1_a &= 0, \notag\\
    (1-1_a)^2 &= 1-1_a, & (1-1_a) + 1_ae^x &= e^{x \cdot 1_a}.
\end{align}

Hence,
\begin{equation}
    R_z(\theta)Z^bX^a = R_z^a(2\theta)X^aZ^bR_z(\theta).
\end{equation}
This shows that the decryption of the rotation gate $R_z(\theta)$ is dependent on the $R_z(2\theta)$ gate.
\end{proof}

\begin{theorem}\label{thrm2:integral_rz}
Any rotation gate $R_z(\pm\frac{\pi}{2^m})$ for $m \in \mathbb{Z^{+}}$ can be exactly decrypted using atmost $m$ rotation gates.
\end{theorem}

\begin{proof}

We start by showing that the decryption of the $R_z(\pm \pi/2)$ gate is trivial and can be performed using the client's Pauli rotations only. Using theorem \ref{thrm1:rz_gate}, we can verify this fact as:
\begin{align}
    R_z(\pm\pi/2)X^aZ^b &= R^a_z(\pm\pi) X^a Z^b R_z(\pm\pi/2), \notag\\
                    &= e^{\mp ia\pi/2}Z^a X^a Z^b R_z(\pm \pi/2).
\label{eq:rz_pi_2}
\end{align}

Using Eq. (\ref{eq:rz_pi_2}), we can recursively decrypt any $R_z(\pm\pi/2^{k})$ with sucessive application of $R_z(\pm\pi/2^{k-1})$ gates untill the $k=1$ i.e. $\pm\pi/2^{k} = \pm\pi/2$.
This shows that the recursive decryption of $R_z(\pm\pi/2^{m})$ where $m \in \mathbb{Z^+}$ requires at most $m$ applications of the $R_z$ gate and at most $m$ rounds of communication between client and server.

\end{proof}

\begin{theorem}\label{thrm3:non_integral_rz}
Any $R_z$ rotation gate with given $\theta$ can be approximately decrypted with arbitrary precision $\epsilon$ using at most $log_2^2(\pi/\epsilon)$ rotation gates.
\end{theorem}

\begin{proof}

We start by showing that any $\theta$ can be approximated by a finite series $S$ with arbitrary precision $\epsilon$ using apprropriate number of elements ($M+1$):
\begin{align}
   \mathcal{S}= \bigg\{ p\pi + \sum_{m=1}^M \frac{a_m\pi}{2^m} \bigg|  a_m\in \{-1, 0,1\}, p \in \mathbb{Z} \bigg\},
\end{align}

Given an arbitrary $\theta$, an integral multiple of $\pi$ can be directly reached using an appropriate value of $p= \theta \text{ (mod }\pi)$.
Now, $\theta_1= \theta-p\pi < \pi$.

Using a modified series, similar to dyadic expression $\sum_{m=1}^\infty \frac{a_m}{2^m} \in [0,1], \forall a_m\in\{0,1\}$, we can say that $\sum_{m=1}^\infty \frac{a_m\pi}{2^m} \in [-\pi,\pi], \forall a_m \in \{-1,0,1\}$.
However, for finite approximation of $\theta_1$, we need to find  $\theta' \in \mathcal{S}$ such that:
\begin{align}
    |\theta_1 - \theta'| < \epsilon.
\end{align}
The precision of the series will be the smallest value achievable in the set $\mathcal{S}$ of $M+1$ elements, i.e., $\pi/2^M$, which implies $|\theta_1-\theta'|=\pi/2^M$.
Hence, the appropriate value of $M$ for approximation with arbitrary precision $\epsilon$ will be:
\begin{align}
    M > log_2(\pi/\epsilon).
\end{align}

Using this limit on the number of elements for finite representation of $\theta$, we now find the upper limit on decryption of arbitrary $R_z(\theta)$ using the following property,
\begin{align}
    R_z(x + y) = R_z(x)R_z(y).
\end{align}

Any arbitrary $R_z(\theta)$ can be perfomed using series $S$ using appropriate values of $p \in\mathbb{Z}, a_m \in \{-1,0,1\}$ for $ m \in \{1,\cdots,M\}$ with $M=\lceil log_2(\pi/\epsilon)\rceil$ as:
\begin{align}
    R_z(\theta) &\approx R_z(p\pi + \sum_{m=1}^M a_m\pi/2^m), \notag \\
     &\approx R_z(p\pi) \prod_{m=1}^M R_z(a_m \pi/2^m).
     \label{eq:rz_decomposed}
\end{align}

Here, $R_z(p\pi)$ can be implemented using at most one $Z$ gate as evident from Eq. (\ref{eq:rz_mpi}), while $R_z(a_m\pi/2^m)$ gates with $m\in \mathbb{Z^{+}}$ requires $m$ rotation gates each as proven in Theorem \ref{thrm2:integral_rz}.
This shows that any arbitrary rotation gate $R_z(\theta)$ can be implemented using at most $M$ integral $R_z(\pm\pi/2^m)$ gates (Eq. (\ref{eq:rz_decomposed})), which in turn requires $m$ rotation gates each for decryption (Theorem \ref{thrm2:integral_rz}).

Hence, the total rotation gates required for recursive decryption of $R_z(\theta)$ is at most $M^2=log_2^2(\pi/\epsilon)$, which is also equal to the number of communication rounds needed between client and server to perform the secure computation.

\end{proof}

\begin{corollary}\label{coro1:efficient_rz}
The decyrption based on an arbitrary $R_z$ rotation is efficient than using a non-parametric gate set whose decomposition is based on the Solovay-Kitaev theorem.
\end{corollary}
\begin{proof}
According to the Solovay-Kitaev theorem, the transformation of arbitrary $R_z(\theta)$ rotation gates to the $H, S, T$ sequence needs a length of  $O(\ln^c(1/\epsilon))$ for $\epsilon$ precision. This requires any blind quantum technique based on such non-parametric gates to consume communication rounds of size $O(\ln^c(1/\epsilon))$, where $c=3.97$ \cite{dawson2006solovay}.

However, Theorem \ref{thrm3:non_integral_rz} shows that an arbitrary rotation gate-based blind approach has an upper bound on communication rounds as $O(\log_2^2(\pi/\epsilon))$.

Hence, showing that a parametric resource set based on $R_z(\theta)$ is more efficient than a non-parametric resource set based on $\{H, S, T\}$ gates for given $\epsilon$ precision as:
\begin{align}
    O(\log_2^2(\pi/\epsilon)) < O(\ln^{3.97}(1/\epsilon)),
\end{align}
\end{proof}

\bibliography{main_ref}


\end{document}